\documentclass{llncs}
\usepackage{amsmath}
\usepackage{mathrsfs}

\usepackage{multirow}
\usepackage{amsfonts}

\usepackage[colorlinks]{hyperref}

\usepackage{setspace}
\usepackage{textcomp}
\usepackage[ruled,linesnumbered]{algorithm2e}
\usepackage{graphicx}
\DeclareGraphicsExtensions{.pdf,.jpg,.png}

\newcommand{\tabincell}[2]{\begin{tabular}{@{}#1@{}}#2\end{tabular}} 

\begin{document}
%

\title{An Algorithm for Reducing Approximate Nearest Neighbor to Approximate Near Neighbor with $O(\log{n})$ Query Time}

\titlerunning{An $O(\log{n})$ Query Time Algorithm for Reducing $\epsilon$-NN to $(c,r)$-NN}
%
\author{Hengzhao Ma\inst{\dag} \and Jianzhong Li\inst{\ddag}
}

\authorrunning{H. Ma, J. Li}
%
\institute{
	Harbin Institute of Technology, Harbin, Heilongjiang 150001, China\\
	\dag\; \email{hzma@stu.hit.edu.cn}\\
	\ddag\; \email{lijzh@hit.edu.cn}
}

\maketitle              
\begin{abstract}
This paper proposes a new algorithm for reducing Approximate Nearest Neighbor problem to Approximate Near Neighbor problem. The advantage of this algorithm is that it achieves  $O(\log{n})$ query time. As a reduction problem, the query time complexity is the times of invoking the algorithm for Approximate Near Neighbor problem. All former algorithms for the same reduction need $polylog(n)$ query time. A box split method proposed by Vaidya is used in our paper to achieve the $O(\log{n})$ query time complexity.

\end{abstract}
\section{Introduction}\label{sec:into}

The approximate nearest neighbor problem, $\epsilon$-NN for short, 
can be defined as follows: 
given a set $P$ of points in a metric space $S$ equipped with a distance function $D$, and a query point $q \in S$, find a point $p\in P$ such that $D(p, q)\le(1+\epsilon) D(p^*, q)$, where $p^*$ has the minimal distance to $q$ in $P$.
$\epsilon$-NN is one of the most important proximity problems 
in computation geometry. 
Many proximity problems in computation geometry can be reduced to $\epsilon$-NN \cite{Goel2001}, such as approximate diameter, approximate furthest neighbor, and so on. 
$\epsilon$-NN is also important in many other areas, such as databases, data mining, information retrieval and machine learning.

Due to its importance, $\epsilon$-NN has been the subject of substantial research efforts. Many algorithms for solving $\epsilon$-NN have been discovered. These works can be summarized into four classes.

The first class of the algorithms tries to build data structures that support solving $\epsilon$-NN efficiently.
Arya et. \cite{Arya1993} give a such algorithm with query time $1/\epsilon^{O(d)}\cdot\log{n}$,
space $1/\epsilon^{O(d)}\cdot n$ and preprocessing time $1/\epsilon^{O(d)}\cdot n\log{n}$.
Another work \cite{Arya1998} gives an algorithm requiring $O(dn)$ space and $O(dn\log{n})$ preprocessing time but query time as high as $(d/\epsilon)^{O(d)}\cdot\log{n}$.
Kleinberg proposes two algorithms in \cite{Kleinberg1997}. 
The first algorithm is deterministic and achieves query time of $O(d\log^2{d}(d+\log{n}))$, using a data structure
that requires $O((n\log{d})^{2d})$ space and $O((n\log{d})^{2d})$ preprocessing time.
The second algorithm is a randomized version of the first one. By a preprocessing procedure that takes $O(d^2\log^2{d}\cdot n\log^2{n})$ time, it reduces the storage requirement to $O(dn\cdot \log^3{n})$, but raises the query time up to $O(n+d\log^3{n})$.

The second class of the algorithms considers the situation of $\epsilon=d^{O(1)}$.  
One such algorithm is given in \cite{Bern1993}. 
It can answer $O(\sqrt{d})$-NN in $O(2^d\log{n})$ time with $O(d8^dn\log{n})$ preprocessing time and $O(d2^dn)$ space.
Chan \cite{Chan1997} improves this result by giving an algorithm that can answer $O(d^{3/2})$-NN in $O(d^2\log{n})$ query time with $O(d^2n\log{n})$ preprocessing time and $O(dn\log{n})$ space.

The third interesting class of work tries to solve $\epsilon$-NN by inspecting some intrinsic dimension of the input point set $P$. 
An exemplar work is in \cite{Krauthgamer2004}. 
The paper gives an algorithm whose query time is bounded by $2^{O(dim(P))}\log{\Delta}+ (1/\epsilon)^{O(dim(P))}$, where $dim(P)$ is the intrinsic dimension of the input point set $P$, and $\Delta$ is the diameter of $P$.

Besides these algorithms mentioned above, Indyk et. \cite{Indyk1998} initiate the work on the fourth class of algorithms. The key idea is to define an Approximate \textit{Near} Neighbor problem, denoted as $(c,r)$-NN, and reduce $\epsilon$-NN to it. The $(c,r)$-NN problem can be viewed as a decisive version of $\epsilon$-NN. The formal definition of $(c,r)$-NN is give in Definition \ref{def:cr-nn} in the next section. 

To use this method to solve $\epsilon$-NN, two parts of problem must be considered. One is how to solve $(c,r)$-NN, and the other is how to reduce $\epsilon$-NN to $(c,r)$-NN. Some works about the two parts of problem are discussed below. Our study focuses on the latter part.

\subsubsection{Algorithms to solve $(c,r)$-NN}
The existing algorithms for $(c,r)$-NN mainly consider the specific situation of $d$-dimensional Euclidean space with $1$-order and $2$-order Minkowski distance metrics. Each input point $x$ is given in the form of $(x_1,\cdots,x_d)$.
And $q$-order Minkowski $L_q$ distance between points $x$ and $y$ is given by $D(x,y)=\left(\sum\limits_{i=1}^d{|x_i-y_i|^q}\right)^{\frac{1}{q}}$. 
The $1$-order and $2$-order Minkowski distance are well-known Manhattan distance and Euclidean distance, respectively.
Another simpler situation, which is the $(c,r)$-NN problem under Hamming cube $\{0,1\}^d$ equipped with Hamming distance, is usually considered in theoretical studies.

Table \ref{tbl:cmp-crnn} summarizes the complexities of the existing algorithms for $(c,r)$-NN under Euclidean space and $L_1$ distance. 
These papers also give solutions under $L_2$ distance, but we omit these results due to space limitation. 
Usually the complexities under $L_2$ distance is higher than that under $L_1$ distance. It is a key characteristic of the existing algorithms for $(c,r)$-NN that they usually have different complexities for problems under different order of Minkowski distance metrics.

\begin{table}[htb]
	\centering
	\caption{Solutions to $(c,r)$-NN under Euclidean space and $L_1$ distance.}
	\label{tbl:cmp-crnn}
	\resizebox{\textwidth}{!}{
		\begin{tabular}{|c|c|c|c|c|c|c|}
			\hline
			\multirow{2}{*}{Source}&\multicolumn{2}{|c|}{Data structure building} &\multicolumn{2}{|c|}{Query} &\multirow{2}{*}{Space}&
			\multirow{2}{*}{\tabincell{c}{Update\\time}}\\
			\cline{2-5}
			&Time&\tabincell{c}{Failure\\probability}
			&Time&\tabincell{c}{Failure\\probability}&&\\
			\hline
			\tabincell{c}{\cite{Indyk1998}\\$(\epsilon=c-1)$}&
			$O(n\cdot\frac{1}{\epsilon^d})$& 0 &
			$O(1)$&0&
			$O(n\cdot\frac{1}{\epsilon^d})$&
			$O(\frac{1}{\epsilon^d})$\\
			\hline
			\hline
			\tabincell{c}{\cite{Kushilevitz2000}\\$(\epsilon=c-1)$}&
			$O\left(n\frac{d^3}{\epsilon^2}(n\log{d})^{O(\frac{1}{\epsilon^2})}\right)$&$O(1)$&
			$O\left(\frac{d}{\epsilon^2}polylog(dn)\cdot\log{\frac{1}{f}}\right)$ &$f$&
			$O\left(\frac{d^3}{\epsilon^2}(n\log{d})^{O(\frac{1}{\epsilon^2})}\right)$&
			$O(n^{O(\frac{1}{\epsilon^2})})$\\
			\hline
			\hline
			\cite{Panigrahy2006}&
			$O(n^{(\frac{c}{c-1})^2}\log{n})$& 0 &
			$O(dn^{o(1)})$&$O(1)$&
			$O(n^{(\frac{c}{c-1})^2})$&
			$O(n^{(\frac{c}{c-1})^2})$\\
			\hline
			\cite{Datar2004,Andoni2014,Andoni2015a}&
			$O(dn^{1+\frac{1}{2c-1}}\log{n})$& 0 &
			$O(dn^{\frac{1}{2c-1}})$&$O(1)$&$O(dn+n^{1+\frac{1}{2c-1}})$&
			$O(dn^{\frac{1}{2c-1}+o(1)})$\\
			\hline
			\cite{Andoni2006}&
			$O(dn^{1+o(1)}\log{n})$& 0 &
			$O(n^{\frac{2c-1}{c^2}})$&$O(1)$&
			$O(dn^{1+o(1)})$&
			$O(dn^{o(1)})$\\
			\hline
		\end{tabular}
	}
\end{table}

The listed solutions in Table \ref{tbl:cmp-crnn} can be divided into three groups. The first group includes the one given in \cite{Indyk1998}, which is deterministic, and the other groups are randomized. The advantage of randomization is that the exponential complexity about $d$ is freed. 
The second group includes the one given in \cite{Kushilevitz2000}, which is based on a random projection method proposed in \cite{Kleinberg1997}. One distinguished characteristic of the method 
is that the data structure building stage is also randomized. 
The last group includes a long line of research work based on Locality Sensitive Hashing (LSH), which is first proposed in \cite{Indyk1998}. These works are summarized into three terms in Table \ref{tbl:cmp-crnn}, which can be viewed as the space-time trade-off under LSH framework. 

Finally, comparing the five results in Table \ref{tbl:cmp-crnn}, it can be seen that the query time grows and the space requirement drops from the first one to the last. The five results form a general space-time trade-off about the solution to $(c,r)$-NN.

\subsubsection{Reducing $\epsilon$-NN to $(c,r)$-NN}
So far there are three different algorithms for such a reduction.
Two of the three algorithms are deterministic \cite{Indyk1998,Har-Peled2001}, and the other one is randomized \cite{Har-Peled2012}. 
The complexities of the three reduction algorithms are summarized in Table \ref{tbl:cmp-reduction}. Note that query time in Table \ref{tbl:cmp-reduction} is the number of invocations of $(c,r)$-NN algorithm. And the preprocessing time about \cite{Indyk1998} is not given because there is no such analysis in that paper.

\begin{table}[htb]
	\centering
	\caption{Comparison of three reductions.}
	\label{tbl:cmp-reduction}
	\resizebox{\textwidth}{!}{
		\begin{tabular}{|c|c|c|c|c|c|c|c|}
			\hline
			\multirow{2}{*}{Source}&\multirow{2}{*}{\tabincell{c}{Approximation\\factor}} &\multicolumn{2}{|c|}{Preprocessing} &\multicolumn{2}{|c|}{Query}
			&\multirow{2}{*}{Space}\\
			\cline{3-6}
			&&Time &\tabincell{c}{Failure\\probability}
			&\tabincell{c}{Time \\(\# of $(c,r)$-NN invoked)}
			&\tabincell{c}{Failure\\probability}& \\
			\hline
			\cite{Har-Peled2012} &\tabincell{c}{$c(1+\gamma)^2$\\$(\gamma\in(\frac{1}{n},\frac{1}{2}))$\\$(c=1+\epsilon)$}&
			$O\left(\frac{T(n,c,f)}{\gamma\log^2{n}}+n\log{n}[Q(n,c,f)+D(n,c,f)]\right)$&$f\log{n}$&
			$O(\log^{O(1)}{n})$& $f\log{n}$&
			$O(\frac{S(n,c,f)}{\gamma\log^2{n}})$\\
			\hline
			\hline
			\cite{Indyk1998} & $1+\epsilon$&
			- &-&
			$O(\log^2{n})$ & 0 &
			$O(n\cdot polylog(n))$\\
			\hline
			\cite{Har-Peled2001} &$1+\epsilon$&
			$O(d\cdot n\frac{\log{n}}{\epsilon}\log{\frac{n}{\epsilon}})$ & 0 &
			$O(\log{\frac{n}{\epsilon}})$ & 0 &
			$O(d\cdot n\frac{\log{n}}{\epsilon}\log{\frac{n}{\epsilon}})$ \\
			\hline
		\end{tabular}
	}
	
\end{table}

Among the three reduction algorithms, the one proposed in \cite{Har-Peled2012} need to be explained in detail.
First, the algorithm outputs a point $p'$ such that $D(q,p')\le c(1+\gamma)^2D(q,p^*)$, where $c=1+\epsilon$ and $p^*$ is the exact NN of $q$. 
Second, the $T(n,c,f)$, $Q(n,c,f)$, $D(n,c,f)$ and $S(n,c,f)$ functions represent the complexity functions of the data structure building time, query time, update time and storage usage for $(c,r)$-NN, respectively.
Third, the parameter $f$ is the failure probability of one $(c,r)$-NN invocation, and is selected so that $f\log{n}$ is a constant less that 1. 

The fourth and the most important point about \cite{Har-Peled2012} is the $O(\log^{O(1)}{n})$ query time. The algorithm given in \cite{Har-Peled2012} explicitly invokes $O(\log{n})$ times of $(c,r)$-NN, and each invocation needs $T(n,c,f)$ time. 
As explained above, the parameter $f$, which is the failure probability of one  $(c,r)$-NN invocation, is set to $O(\frac{1}{\log{n}})$. 
Note that the algorithms for $(c,r)$-NN given in Table \ref{tbl:cmp-crnn} all have constant failure probability\footnote{The deterministic one has exponential dependence on $d$, so it it rarely used in theory and practice.}. 
In order to satisfy the requirement of $O(\frac{1}{\log{n}})$ failure probability of one $(c,r)$-NN invocation, each time the algorithm in \cite{Har-Peled2012} invokes $(c,r)$-NN, the algorithms for $(c,r)$-NN with constant failure probability must be executed multiple times, which is $O(\log^{O(1)}{n})$ times in expectation. 
Multiplying $O(\log{n})$ invocations of $(c,r)$-NN and $O(\log^{O(1)}{n})$ executions of $(c,r)$-NN algorithm for each invocation, we obtain that the algorithm in \cite{Har-Peled2012} actually invokes $O(\log^{O(1)}{n})$ times of $(c,r)$-NN algorithm. This observation is confirmed in \cite{Andoni2017a}.

\subsubsection{Our method}
We propose a new algorithm in this paper for reducing $\epsilon$-NN to $(c,r)$-NN.
Comparing with the former works \cite{Har-Peled2012,Indyk1998,Har-Peled2001}, our algorithm has the following characteristics:

(1) It achieves $O(\log{n})$ query time, counted in the number of invocations of $(c,r)$-NN algorithm. It is superior to all the other three works. This is the most distinguished contribution of this paper.

(2) Its preprocessing time is $O((\frac{d}{\epsilon})^d\cdot n\log{n})$, and the space complexity is $O((\frac{d}{\epsilon})^d\cdot n)$. Our method has better complexity than the other three works in terms of $n$, so that it is much suitable to big data with low or fixed dimensionality. This situation is plausible in many applications like road-networks and so on.

(3) In terms of the parameterized complexity treating $d$ as a constant, our result is the closest to the well recognized \emph{optimal} complexity claimed in \cite{Arya1998}, which requires $O(n\log{n})$ preprocessing time, $O(n)$ space and $O(\log{n})$ query time. 

Note that there is an $O((d/\epsilon)^d)$ factor in our preprocessing and space complexity. This factor originates from a lemma we used in \cite{Vaidya1986}.
We point out that the upper bound $O((d/\epsilon)^d)$ is actually very loose. 
There really is possibility to reduce the upper bound, and thus make our result more close to optimal. 
In this sense, our work is more promising than all the other three works. However, reducing the upper bound $O((d/\epsilon)^d)$ is out of this paper's scope, and is left as our future work.

\section{Problem Definitions and Mathematical Preparations}\label{sec:math-prep}

\subsection{Problem definitions}
We focus on $\epsilon$-NN in euclidean space $R^d$. The input is a set $P$ of $n$ points extracted from $R^d$ and a distance metric $L_q$. Each point $x$ is given as the form $(x_1,\cdots,x_d)$. $L_q$ distance metric between points $x$ and $y$ is given by $D(x,y)=\left(\sum\limits_{i=1}^d{|x_i-y_i|^q}\right)^{\frac{1}{q}}$. 

Denote $B(p,r)$ to be the $d$-dimensional ball centered at $p$ and with radius $r$. And let $p'\in B(p,r)$ be equivalent to $D(p',p)\le r$. We first give the definitions of $\epsilon$-NN and $(c,r)$-NN problems.

\begin{definition}[$\epsilon$-NN]\label{def:eps-nn}
	Given a set $P$ of points extracted from $R^d$, a query point $q\in R^d$, and an approximation factor $\epsilon$, find a point $p'\subseteq P$ such that $D(p',q)\leq(1+\epsilon)D(p^*,q)$ 
	where $D(p^*,q)=\min\limits_{p\in P}\{D(p,q)\}$.
\end{definition}

\begin{remark}
	$p^*$ is called the nearest neighbor (NN), or exact NN to $q$, and $p'$ is called an $\epsilon$-NN to $q$.
\end{remark}

\begin{definition}[$(c,r)$-NN]\label{def:cr-nn}
	Given a set $P$ of points extracted from $R^d$, a query point $q\in R^d$, a query range $r$, and an approximation factor $c>1$, $(c,r)$-NN problem is to design an algorithm satisfying these: 
	\begin{enumerate}
		\item if there is a point $p_0\in P$ satisfying $p_0\in B(q,r)$, then return a point $p' \in P$ such that $p'\in B(q,c \cdot r)$;
		\item if $D(p,q) > c\cdot r$ for $\forall p \in P$, then return No.
	\end{enumerate}
\end{definition}

\begin{remark}
	There are multiple names referring to the same problem defined above. In the papers related to LSH, it is referred as $(c,r)$-NN. In \cite{Indyk1998}, it is called approximate Point Location in Equal Balls, which is denoted as $\epsilon$-PLEB where $\epsilon=c-1$. In more recent papers like \cite{Har-Peled2001}, it is called Approximate Near Neighbor problem. 	
\end{remark}

Next we give the definition of the reduction problem to be solved in this paper, i.e., the problem of reducing $\epsilon$-NN to $(c,r)$-NN.

\begin{definition}[Reduction Problem]\label{def:reduction}
	Given a set $P$ of points extracted from $R^d$, a query point $q\in R^d$, an approximation factor $\epsilon$, and an algorithm $\mathcal{A}$ for $(c,r)$-NN, the reduction problem is to find an $\epsilon$-NN to $q$ by invoking the algorithm $\mathcal{A}$ as an oracle.
\end{definition}

\begin{remark}
	To solve the reduction problem, a preprocessing phase is usually needed, which is to devise a data structure $\mathcal{D}$ based on $P$. Thus the problem of reducing $\epsilon$-NN to $(c,r)$-NN is divided into two phases. The first is data structure building phase, or preprocessing phase. The second is $\epsilon$-NN searching phase, or query phase. The $(c,r)$-NN algorithm $\mathcal{A}$ is invoked in query phase as an oracle, which characterizes the algorithm as a Turing reduction from $\epsilon$-NN to $(c,r)$-NN. 
	
	The time complexity of the algorithm for the reduction problem consists of two parts, namely, preprocessing time complexity and query time complexity. An important note is that the query time complexity is the number of invocations of $(c,r)$-NN algorithm $\mathcal{A}$. This is the well recognized method for analyze the time complexity of a Turing reduction.
\end{remark}

\subsection{Mathematical Preparations}\label{subsec:def-lema}
In this section we introduce some denotations and lemmas to support the idea of our algorithm for reducing $\epsilon$-NN to $(c,r)$-NN.

\subsubsection{Denotations}
Define a box $\mathfrak{b}$ in $R^d$ to be the product of $d$ intervals, i.e., $I_1\times I_2\times\cdots\times I_d$ where $I_i$ is either open, closed or semi-closed interval, $1\le i\le m$. 
A box is cubical iff all the $d$ intervals defining the box are of the same length. The side length a cubical box, which is the length of any interval defining the cubical box, is denoted as $len(\mathfrak{b})$. A minimal cubical box (MCB) for a point set $P$, denoted as $MCB(P)$, is the cubical box containing all the points in $P$ and has the minimal side length. Note that $MCB(P)$ may not be unique.

Given a point set $P$ and a box $\mathfrak{b}$, let $p\in \mathfrak{b}$ denote that a point $p\in P$ falls inside box $\mathfrak{b}$, and let $|\mathfrak{b}\cap P|$ denote the number of points in $P$ that falls inside $\mathfrak{b}$. We will use $|\mathfrak{b}|$ for short, if not causing ambiguity.

Given a collection of MCBs $\mathcal{B}=\{\mathfrak{b}_1,\cdots,\mathfrak{b}_m\}$, define $D_{max}(\mathfrak{b})$, $D_{min}(\mathfrak{b},\mathfrak{b}')$, $D_{max}(\mathfrak{b}, \mathfrak{b}')$ as follows:

$$D_{max}(\mathfrak{b})=\max\limits_{p_1,p_2\in \mathfrak{b}}{D(p_1,p_2)},\forall \mathfrak{b}\in \mathcal{B}$$
$$D_{min}(\mathfrak{b}=\mathfrak{b}')=\min\limits_{p\in \mathfrak{b},p'\in \mathfrak{b}'}\{D(p,p')\},
D_{max}(\mathfrak{b}, \mathfrak{b}')=\max\limits_{p\in \mathfrak{b},p'\in \mathfrak{b}'}\{D(p,p')\}, \forall \mathfrak{b},\mathfrak{b}'\in \mathcal{B}$$

With the above denotations, define $Est(\mathfrak{b})$ as follows:
\begin{equation}\label{eqtn:est(b)}
Est(\mathfrak{b})= \left\{
\begin{aligned}
D_{max}(\mathfrak{b}) &,& if\;|\mathfrak{b}\cap P|\ge 2\\
\min\limits_{\mathfrak{b}'\in Nbr(\mathfrak{b})} \{D_{max}(\mathfrak{b},\mathfrak{b}')\}&,& otherwise
\end{aligned}
\right.
\end{equation}
where $Nbr(\mathfrak{b}) = \left\{\mathfrak{b}'\mid{D_{min} (\mathfrak{b}, \mathfrak{b}')\le r}\right\}$, and the parameter $r$ should satisfy $r\ge Est(\mathfrak{b})$.\footnote{It can be verified that, as long as $r\ge Est(\mathfrak{b})$ is satisfied, the value of $r$ doesn't influence the value of $Est(\mathfrak{b})$. The specific value of $r$ will be shown latter.}

For an MCB $\mathfrak{b}$, we associate a ball with it.
Pick an arbitrary point $c_\mathfrak{b}\in \mathfrak{b}$, and let $r_\mathfrak{b}=Est(\mathfrak{b})$, then we have a ball $B(c_\mathfrak{b},r_\mathfrak{b})$. It is easily verified that every point in $\mathfrak{b}$ is within a distance of $Est(\mathfrak{b})$ from $c_\mathfrak{b}$, in another way to say, the ball $B(c_\mathfrak{b}, r_\mathfrak{b})$ encloses every point in $\mathfrak{b}$. 
We call $B(c_\mathfrak{b}, r_\mathfrak{b})$ the \emph{enclosing ball} for box $\mathfrak{b}$.

Next we start to introduce the lemmas while discussing different situations of $\epsilon$-NN search. In the following discussion, we will assume that we have an MCB $\mathfrak{b}$ of the input point set $P$, an enclosing ball $B(c_\mathfrak{b}, r_\mathfrak{b})$ of the MCB $\mathfrak{b}$, and a query point $q$. 

\subsubsection{Situation 1}
The first and an easy situation is that, if $q$ is far enough from $c_{\mathfrak{b}}$ then every point in $\mathfrak{b}$ is an $\epsilon$-NN to $q$.
The following value $T_1(\mathfrak{b})$ explains the threshold for \emph{far enough}, and Lemma \ref{lema:t1} depicts the situation discussed above.

\begin{definition}\label{def:t1}
	For an MCB of a point set $P$, define $T_1(\mathfrak{b}) =(1+2/\epsilon)r_\mathfrak{b}$. 
\end{definition}

\begin{lemma}\label{lema:t1}
	If $D(q,c_\mathfrak{b})\ge T_1(\mathfrak{b})$, then every point in $\mathfrak{b}$ is an $\epsilon$-NN to $q$.
\end{lemma}

\begin{proof}
	If $|\mathfrak{b}|=1$ then the lemma is trivial. 
	Assume $|\mathfrak{b}|\ge 2$. 	
	Starting from the given condition, we first prove $(1+\epsilon) (D(q,c_\mathfrak{b})-r_\mathfrak{b})\ge D(q,c_\mathfrak{b})+r_\mathfrak{b}$ as follows:
	\[\begin{split}
	D(q,c_\mathfrak{b})&\ge (1+2/\epsilon)r_\mathfrak{b} \Rightarrow\\
	\epsilon\cdot D(q,c_\mathfrak{b})&\ge (\epsilon+2)r_\mathfrak{b} \Rightarrow\\
	(1+\epsilon)D(q,c_\mathfrak{b})&\ge D(q,c_\mathfrak{b})+(\epsilon+2)r_\mathfrak{b} \Rightarrow\\
	(1+\epsilon) (D(q,c_\mathfrak{b})-r_\mathfrak{b})&\ge  D(q,c_\mathfrak{b})+r_\mathfrak{b}.\\
	\end{split}\]
	
	Let the minimal distance from the query point $q$ to a point in $\mathfrak{b}$ be $D(q,\mathfrak{b})$. Clearly, we have $D(q,\mathfrak{b})\ge D(q,c_\mathfrak{b})-r_\mathfrak{b}$. Then we have $D(q,p)\le D(q,c_\mathfrak{b})+r_\mathfrak{b} \le (1+\epsilon)(D(q,c_\mathfrak{b})-r_\mathfrak{b})\le (1+\epsilon)D(q,\mathfrak{b})$ for $\forall p\in \mathfrak{b}$. This indicates that every point $p\in \mathfrak{b}$ is $\epsilon$-NN to q.
	\qed
\end{proof}	

\subsubsection{Situation 2}
If $q$ is not as far from $c_\mathfrak{b}$ as a distance of $T_1(\mathfrak{b})$, i.e., $D(q,c_{\mathfrak{b}}) < T_1(\mathfrak{b})$, then we split $\mathfrak{b}$ into a set of sub-boxes $\{\mathfrak{b}_1,\cdots, \mathfrak{b}_m\}$, and calculate the enclosing balls $B(c_{\mathfrak{b}_i},r_{\mathfrak{b}_i})$ for each box $\mathfrak{b}_i, 1\le i\le m$. The next situation is that if $q$ is still far enough from each point in $\{c_{\mathfrak{b}_1},\cdots, c_{\mathfrak{b}_m}\}$, i.e., the centers of the enclosing balls, then we can still tell that every point in $\mathfrak{b}$ is an $\epsilon$-NN to $q$. We give another threshold $T_2(\mathfrak{b})$ based on this idea, and formalize the idea into Lemma \ref{lema:t2}.This lemma also discusses the quantitative relationship between $T_2(\mathfrak{b})$ and $T_1(\mathfrak{b})$.

\begin{definition}\label{def:rmax-box}
	For an MCB of a point set $P$, split $\mathfrak{b}$ into a set of sub-boxes $\{\mathfrak{b}_1,\cdots, \mathfrak{b}_m\}$. Each of these sub-boxes is an MCB of a point set $P'\subset P$. Then let
	$B(c_{\mathfrak{b}_i},r_{\mathfrak{b}_i})$ be the enclosing ball of sub-box $\mathfrak{b}_i$, $1\leq i \leq m$.
	Define $rmax_\mathfrak{b} = \max\limits_{i}{\{r_{\mathfrak{b}_i}\}}$. In case of $|\mathfrak{b}|=1$, let $rmax_\mathfrak{b}=0$.
\end{definition}

\begin{definition}\label{def:t2}
	Define $T_2(\mathfrak{b})=r_\mathfrak{b} +(1+2/\epsilon)rmax_\mathfrak{b}$.
\end{definition}

\begin{lemma}\label{lema:t2}
	We have the following statements:
	\begin{enumerate}
		\item if $D(q,c_\mathfrak{b})\ge T_2(\mathfrak{b})$, then every point in $\mathfrak{b}$ is $\epsilon$-NN to $q$;
		\item if $rmax_\mathfrak{b} < \frac{2}{2+\epsilon}r_\mathfrak{b}$, then $T_2 (\mathfrak{b}) < T_1(\mathfrak{b})$.
	\end{enumerate}
\end{lemma}

\begin{proof}
	For the first statement, if $|\mathfrak{b}|=1$ then it is trivial. Assume $|\mathfrak{b}|\ge 2$. Since the center of the enclosing ball of the box $\mathfrak{b}$ is chosen as any point in $\mathfrak{b}$, it is easy to see that $c_{\mathfrak{b}_i} \in \mathfrak{b}$ for each sub-box $\mathfrak{b}_i$. 
	This in turn indicates that $D(c_{\mathfrak{b}_i},c_\mathfrak{b})\le r_\mathfrak{b}$. Thus, if $D(q,c_\mathfrak{b})\ge T_2(\mathfrak{b})$, we have $D(q,c_{\mathfrak{b}_i})\ge D(q,c_\mathfrak{b}) - D(c_\mathfrak{b},c_{\mathfrak{b}_i}) \ge T_2(\mathfrak{b})-r_\mathfrak{b} =(1+2/\epsilon)rmax_\mathfrak{b}\ge (1+2/\epsilon)r_{\mathfrak{b}_i}= T_1(\mathfrak{b}_i)$. 
	According to Lemma \ref{lema:t1} we know that every point in $\mathfrak{b}_i$ is $\epsilon$-NN to $q$. 
	Since the subscript $i$ is arbitrary in $[1,m]$, we conclude that every point in $\mathfrak{b}$ is $\epsilon$-NN to $q$.
	
	The second statement can be easily verified, and the proof is omitted here.
	\qed
\end{proof}

\subsubsection{Situation 3}
If $q$ is still not as far from $c_\mathfrak{b}$ as a distance of $T_2(\mathfrak{b})$, it is time to ask the algorithm of $(c,r)$-NN for help.
Let $\mathcal{A}(Q,\mathfrak{q},c,r)$ be any algorithm for solving $(c,r)$-NN, where $Q$ is the input point set, $\mathfrak{q}$ is the query point, $r$ is the query range, and $c$ is the approximation factor. The meanings of these four parameters are already given in Definition \ref{def:cr-nn}. The goal of invoking $\mathcal{A}$ is that, if $\mathcal{A}$ answers $No$ then still every point in $\mathfrak{b}$ is an $\epsilon$-NN to $q$. The following lemma shows how to set the four input parameters to fulfill the goal.

\begin{lemma}\label{lema:crnn-param}
	Let $\mathcal{A}(Q,\mathfrak{q},c,r)$ be any algorithm for $(c,r)$-NN. We have the following statements:
	\begin{enumerate}
		\item if we set $Q=\{c_{\mathfrak{b}_1},\cdots c_{\mathfrak{b}_m}\}$,  $\mathfrak{q}=q$, $r=\max \limits_{i} \{T_2 (\mathfrak{b}_i)\}$, and let $c$ satisfy $c\cdot r= \max \limits_{i}\{T_1 (\mathfrak{b}_i)\}$, and invoke $\mathcal{A}(Q,\mathfrak{q},c,r)$, then if $\mathcal{A}$ returns $No$, we can pick any point in $\mathfrak{b}$ as the answer of $\epsilon$-NN to $q$;
		\item if $rmax_{\mathfrak{b}_i} <\frac{2}{2+\epsilon}r_{\mathfrak{b}_i}$ holds for each $\mathfrak{b}_i$, $1\le i\le m$, then our settings for $c$ and $r$ satisfy the requirement of $(c,r)$-NN problem definition. i.e. $c>1$.
	\end{enumerate}
\end{lemma}

\begin{proof} 	
	For the first statement, according to Definition \ref{def:cr-nn} of $(c,r)$-NN problem, if there exists a point in $\{c_{\mathfrak{b}_1},\cdots,c_{\mathfrak{b}_m}\}$ lying inside $B(q,r)$ where $r=\max \limits_{i} \{T_2 (\mathfrak{b}_i)\}$, then $\mathcal{A}$ will return some point $c_{\mathfrak{b}_j}$ such that $c_{\mathfrak{b}_j}\in B(q,cr)$ where $c\cdot r=\max \limits_{i} \{T_1 (\mathfrak{b}_i)\}$. If all points in $\{c_{\mathfrak{b}_1},\cdots,c_{\mathfrak{b}_m}\}$ are outside $B(q,cr)$, then $\mathcal{A}$ will return $No$. If the minimal distance from $q$ to $\{c_{\mathfrak{b}_1}, \cdots,c_{\mathfrak{b}_m}\}$ falls in the undefined range $[r,cr]$, $\mathcal{A}$ will return either $No$ or a point $c_{\mathfrak{b}_j}$ such that $r\le D(q,c_{\mathfrak{b}_j})\le cr$.
	
	On the other hand, according to Lemma \ref{lema:t1} and \ref{lema:t2}, if all the points in $\{c_{\mathfrak{b}_1},\cdots,c_{\mathfrak{b}_m}\}$ satisfy  $D(q,c_{\mathfrak{b}_i})\ge r=\max\limits_{i} \{T_2 (\mathfrak{b}_i)\}\ge T_2(\mathfrak{b}_i)$, then all points in all $\mathfrak{b}_i$ are $\epsilon$-NN to $q$, $1\le i\le m$, and these are already all points in $\mathfrak{b}$.

	Combining the two parts of analysis, we can conclude that if $\mathcal{A}$ returns $No$, it must be the situation that the minimal distance from $q$ to $\{c_{\mathfrak{b}_1}, \cdots,c_{\mathfrak{b}_m}\}$ is not less than $r=\max \limits_{i} \{T_2 (\mathfrak{b}_i)\}$. Equivalently, $D(q,c_{\mathfrak{b}_i})\ge \max \limits_{i} \{T_2 (\mathfrak{b}_i)\}\ge T_2(\mathfrak{b}_i)$ holds for each box $\mathfrak{b}_i$, $1\le i\le m$. Then according to Lemma \ref{lema:t2}, every point in $\mathfrak{b}_i$ is an $\epsilon$-NN to $q$. Since the subscript $i$ is arbitrary in $[1,m]$, we conclude that all points in $\mathfrak{b}$ are $\epsilon$-NN to $q$.

	For the second statement, if $rmax_{\mathfrak{b}_i} <\frac{2}{2+\epsilon}r_{\mathfrak{b}_i}$, then according to Lemma \ref{lema:t2}, $T_2(\mathfrak{b}_i) < T_1(\mathfrak{b}_i)$, $1\le i\le m$. Taking maximum on both sides of the inequality, we have $\max\limits_{i} \{T_2 (\mathfrak{b}_i)\}<\max\limits_{i} \{T_1 (\mathfrak{b}_i)\}$. Since we set $r=\max \limits_{i} \{T_2 (\mathfrak{b}_i)\}$, and let $c$ satisfy $c\cdot r= \max \limits_{i}\{T_1 (\mathfrak{b}_i)\}$, we have $r<cr$ which induces that $c>1$. Then  the proof is done.
	\qed
\end{proof}

\subsubsection{Situation 4}
As what is said in Lemma \ref{lema:crnn-param}, if the algorithm $\mathcal{A}$ returns $No$ then the search of $\epsilon$-NN terminates with returning an arbitrary point in $\mathfrak{b}$. According to Definition \ref{def:cr-nn}, $\mathcal{A}$ can also return some point $c_{\mathfrak{b}_i}\in Q$ other than $No$. In that case the search must continues. At first glance, the same procedure should be recursively carried out, by applying Lemma \ref{lema:t1}, \ref{lema:t2}, \ref{lema:crnn-param} one by one on box $\mathfrak{b}_i$, where the point $c_{\mathfrak{b}_i}$ returned by $\mathcal{A}$ is the center of the enclosing ball of box $\mathfrak{b}_i$. 
However, to guarantee that the algorithm returns a correct $\epsilon$-NN, the box considered by the algorithm must encloses the exact NN $p^*$. But the box $\mathfrak{b}_i$ may not enclose $p^*$, which would ruin the correctness of the algorithm. 
Thus, we need to expand the search range to the boxes near to $\mathfrak{b}_i$.
The following Lemma \ref{lema:nn-in-nbr} gives the bounds of the search range and ensures that $p^*$ lies in the range.

\begin{definition}\label{def:rmax-collection}
	For a collection of MCBs $\mathcal{B}=\{\mathfrak{b}_1,\cdots, \mathfrak{b}_m\}$, let $rmax_{\mathfrak{b}_i}$ be defined as Definition \ref{def:rmax-box} for each $\mathfrak{b}_i$, $1\le i\le m$. Then define $rmax_{\mathcal{B}}=\max\limits_{\mathfrak{b}_i\in \mathcal{B}}\{rmax_{\mathfrak{b}_i}\}$.  
\end{definition}

\begin{definition}\label{def:nbr}
	Define $Nbr(\mathfrak{b})$ as
	$$\left\{\mathfrak{b}'\in\mathcal{B} \mid {D(c_{\mathfrak{b}'},c_{\mathfrak{b}}) \le (3+4\epsilon)rmax_{\mathcal{B}_{\mathfrak{b}_\mathrm{s}}}} \right\}$$ where $\mathcal{B}_{\mathfrak{b}_\mathrm{s}}= Nbr(\mathfrak{b}_\mathrm{s})$ and $\mathfrak{b}_\mathrm{s}$ is the super box of $\mathfrak{b}$.
\end{definition}

\begin{remark}
	The definition of $Nbr$ sets is a \emph{recursive definition}. For a box $\mathfrak{b}$, its $Nbr(\mathfrak{b})$ set is defined based on the $Nbr(\mathfrak{b}_s)$ set of its super box $\mathfrak{b}_s$. It requires that the boxes are recursively split, which can be represented as a tree structure. The formal description of the tree structure is given in Section \ref{subsec:preprocessing}.
\end{remark}

\begin{lemma}\label{lema:nn-in-nbr}
	Given the query point $q$, and a collection of boxes $\{\mathfrak{b}_1,\cdots,\mathfrak{b}_m\}$, if we find a box $\mathfrak{b}_i$ satisfying $D(q, c_{\mathfrak{b}_i})\le \max\limits_{i} \{T_1 (\mathfrak{b}_i)\}$, then the nearest neighbor of $q$ lies in and can only lie in $Nbr(\mathfrak{b}_i)$, i.e., $p^*\in Nbr(\mathfrak{b}_i)$.
\end{lemma}

\begin{proof}
	We have known that $D(q,c_{\mathfrak{b}_i})\le\max\limits_{i}\{T_1 (\mathfrak{b}_i)\}$. By the definition of the enclosing ball of box $\mathfrak{b}$, the center of the ball is an arbitrary point picked from $\mathfrak{b}$. Then we have $c_{\mathfrak{b}_i}\in \mathfrak{b}_i$ for box $\mathfrak{b}_i$. Furthermore, since $\mathfrak{b}_i$ is a sub-box of $\mathfrak{b}$, then $c_{\mathfrak{b}_i}\in \mathfrak{b}$ too. 
	Thus $D(q,c_{\mathfrak{b}_i})$ can serve as an upper bound of the distance from  $q$ to its nearest neighbor. Let $p^*$ denote the exact nearest neighbor of $q$ in $\mathfrak{b}$, and we have $D(q,p^*)\le D(q,c_{\mathfrak{b}_i})\le \max\limits_{i} \{T_1 (\mathfrak{b}_i)\}$. Further we have $\max\limits_{i} \{T_1 (\mathfrak{b}_i)\}=\max\limits_{i} \{(1+2/\epsilon)r_{\mathfrak{b}_i}\} =(1+2/\epsilon)rmax_\mathfrak{b}$. Then $D(q,p^*)\le (1+2/\epsilon)rmax_\mathfrak{b}$.
	
	If $p^*$ lies in some box $\mathfrak{b}_j$, we prove $D(c_{\mathfrak{b}_i},c_{\mathfrak{b}_j})\le (3+4/\epsilon)rmax_\mathfrak{b}$. If $i=j$ then this trivially holds. If not, first we have $D(p^*, c_{\mathfrak{b}_j}) \le r_{\mathfrak{b}_j} \le rmax_\mathfrak{b}$ since $p^*\in \mathfrak{b}_j$. Thus, $D(c_{\mathfrak{b}_i},c_{\mathfrak{b}_j})\le D(c_{\mathfrak{b}_i},q) + D(q,p^*)+D(p^*,c_{\mathfrak{b}_j})\le (1+2/\epsilon)rmax_\mathfrak{b} +(1+2/\epsilon)rmax_\mathfrak{b}+rmax_\mathfrak{b}= (3+4/\epsilon)rmax_\mathfrak{b}$. This indicates that $\mathfrak{b}_j\in Nbr(\mathfrak{b}_i)$.
	
	For a box $\mathfrak{b}_k$ out of $Nbr(\mathfrak{b}_i)$, i.e. $D(c_{\mathfrak{b}_k},c_{\mathfrak{b}_i}) > (3+4/\epsilon)rmax_\mathfrak{b}$,  suppose in contrary that the nearest neighbor $p^*\in \mathfrak{b}_k$, which indicates $D(p^*,c_{\mathfrak{b}_k})\le r_{\mathfrak{b}_k}$. Then we have $D(q,p^*)\ge D(c_{\mathfrak{b}_k},c_{\mathfrak{b}_i}) - D(c_{\mathfrak{b}_i},q) - D(c_{\mathfrak{b}_k},p^*)> (3+4/\epsilon)rmax_\mathfrak{b} - (1+2/\epsilon)rmax_\mathfrak{b}-r_{\mathfrak{b}_k}>(1+2/\epsilon)rmax_\mathfrak{b}$. This conflicts with the conclusion we get above. 
	
	So far both the sufficient and necessary conditions are proved, and the proof is done.
	\qed
\end{proof}

We are done introducing the mathematical preparations. In the next section we will propose our algorithm based on the lemmas given above.

\section{Algorithms}\label{sec:alg}
In this section we propose our algorithm for reducing $\epsilon$-NN to $(c,r)$-NN, including the preprocessing and query algorithm.

\subsection{Preprocessing}\label{subsec:preprocessing}
Our preprocessing algorithm mainly consists of two sub-procedures. One is to build the box split tree $T$, and the other is to construct the $Nbr$ sets.

\subsubsection{Building the box split tree}
We first give the definition of the box split tree.

\begin{definition}[Box split tree]\label{def:split-tree}
	Given a point set $P$ and its MCB $\mathfrak{b}_P$, a tree $T$ is a \emph{box split tree} iff:
	\begin{enumerate}
		\item the root of $T$ is $\mathfrak{b}_P$;
		\item each non-root node of $T$ is an MCB of a point set $P'\subset P$;
		\item if box $\mathfrak{b}'$ is a sub-box of $\mathfrak{b}$, then there is an edge between the node for $\mathfrak{b}$ and the node for $\mathfrak{b}'$ in $T$;
		\item each node has at least $2$ child nodes, and at most $|P|$ child nodes;
		\item $rmax_{\mathfrak{b}}< \frac{2}{2+\epsilon} r_{\mathfrak{b}}$ holds for each box $\mathfrak{b}$ in $T$.
	\end{enumerate}
	Further, $T$ is fully built iff each box at the leaf nodes of $T$ contains only one point.
\end{definition}

\begin{remark}
	The fifth term is required by the second statement of Lemma \ref{lema:crnn-param}.
\end{remark}

We use a box split method to build the box split tree. This method is originally proposed in \cite{Vaidya1986}, and also used in several other papers \cite{Feder1988,Callahan1995}. It starts from the MCB $\mathfrak{b}_P$ of the point set $P$, then continuously splits $\mathfrak{b}_P$ into a collection $\mathcal{B}$ of cubical boxes until each box in $\mathcal{B}$ contains only one point. The whole produce of the method takes $O(dn\log{n})$ time where $n=|P|$. The method proceeds in a series of split steps. In each split step, the box $\mathfrak{b}_L$ with the largest side length in the current collection $\mathcal{B}$ is chosen and split. Define $h_i(\mathfrak{b})$ to be the hyperplane orthogonal to the i-th coordinate axis and passing through the center of $\mathfrak{b}$. One split step will split $\mathfrak{b}_L$ into at most $2^d$ sub-boxes using all $h_i(\mathfrak{b}_L)$, each of which is an MCB. The set of non-empty sub-boxes generated by conducting one split step on $\mathfrak{b}$ is denoted as $Succ(\mathfrak{b})$.
The details of the box splitting method can be found in \cite{Vaidya1986}. 

Next we describe how to use this method to build the box split tree $T$. The main obstacle is to satisfy the fifth term in Definition \ref{def:split-tree}, i.e., $rmax_\mathfrak{b}< \frac{2}{2+\epsilon}r_\mathfrak{b}$ for each box $\mathfrak{b}$ in $T$. We use the following techniques to solve this problem. 

When a split step is executed and a box $\mathfrak{b}$ is split, we temporarily store the sub-boxes of $\mathfrak{b}$ in a max-heap $H_\mathfrak{b}$, which is ordered on the side length of the boxes in the heap.
Recall the definitions in Section \ref{subsec:def-lema}, the side length of a box $\mathfrak{b}$ is denoted as $len(\mathfrak{b})$.
When box $\mathfrak{b}$ is split fine enough so that $rmax_\mathfrak{b}< \frac{2}{2+\epsilon}r_\mathfrak{b}$ is satisfied, the algorithm will create a node for each $\mathfrak{b}'\in H_\mathfrak{b}$, and hang it under the node for box $\mathfrak{b}$ in the box split tree $T$. Then for each $\mathfrak{b}'$ at these newly created leaf nodes, a max-heap is created to store its sub-boxes. 
In an overview, a max-heap is maintained for each box at the leaf nodes of the box split tree.

In each split step, the box with the largest volume is split. To efficiently pick out this box, a secondary heap $H_2$ is maintained. The heaps for the leaf nodes are called the primary heaps in contrast. The elements in $H_2$ is just the top elements in each primary heap, together with a pointer to its corresponding primary heap. Apparently the top element $\mathfrak{b}_{top}$ in $H_2$ is the box with largest volume. When $\mathfrak{b}_{top}$ is picked, the primary and secondary heap will pop it out simultaneously. 
Then $\mathfrak{b}_{top}$ is split by conducting one split step on it, generating $Succ(\mathfrak{b})$. These sub-boxes in $Succ(\mathfrak{b})$ will be added into the primary heap where $\mathfrak{b}_{top}$ formerly resides. 
When this primary heap finishes maintaining, its top element is inserted into the secondary heap. And then the iteration continues. 

We point out the last problem to solve in order to satisfy the fifth term in Definition \ref{def:split-tree}. The heaps, including the primary heaps and the secondary heap, are organized according to the $len$ value of the boxes, in order to retrieve the box with the largest volume. On the other hand, the condition of $rmax_\mathfrak{b}<\frac{2}{2+\epsilon}r_\mathfrak{b}$ is based on the $Est$ value of the boxes, because here we have $rmax_\mathfrak{b}=\max\limits_{b'\in H_\mathfrak{b}}{\{Est(\mathfrak{b}')\}}$. 
Notice that the top element $\mathfrak{b}_{top}$ in the primary heap have the largest $len$ value, but may not has the largest $Est$ value. 
So we can not directly check $r_{\mathfrak{b}_{top}}< \frac{2}{2+\epsilon}r_\mathfrak{b}$ to decide whether $\mathfrak{b}$ is split fine enough. 
Fortunately, the $len$ and $Est$ value of a box have certain quantity relationships, which is formalized into the following lemma. 

\begin{lemma}\label{lema:len-est}
	For the MCB $\mathfrak{b}$ of any point set $P$ where $|\mathfrak{b}|\ge 2$, we have $len(\mathfrak{b})\le Est(\mathfrak{b})\le d\cdot len(\mathfrak{b})$. In the situation that $|\mathfrak{b}|=1$, we redefine $len(\mathfrak{b})$ as $len(\mathfrak{b})=Est(\mathfrak{b})$ to make this inequality consistent.
\end{lemma}

\begin{proof}
	The lemma already fixes the situation of $|\mathfrak{b}|=1$, and thus the proof focuses on when $|\mathfrak{b}|\ge 2$.
	
	If $len(\mathfrak{b})>Est(\mathfrak{b})$, then $\mathfrak{b}$ can be shrunk and still contain all points in $\mathfrak{b}$, which conflicts with that $\mathfrak{b}$ is the Minimal Cubical Box (MCB) of $P$. Thus $len(\mathfrak{b})\le Est(\mathfrak{b})$.
	
	Recalling Equation \ref{eqtn:est(b)}, when $|\mathfrak{b}|\ge 2$, $Est(\mathfrak{b})$ is defined to be $D_{max}(\mathfrak{b})$. Note that the $L_q$ distance between two points in $d$-dimensional space is bounded by $d$ times of the $L_\infty$ distance between them. Since the points are enclosed by box $\mathfrak{b}$ whose side length is $len(\mathfrak{b})$, we conclude that $Est(\mathfrak{b})\le d\cdot len(\mathfrak{b})$.
	
	Both sides of the inequality is proved.
	\qed
	
\end{proof}

With the help of Lemma \ref{lema:len-est}, we have the following Lemma \ref{lema:split-fine-enough} about the criteria for deciding whether a box is split fine enough.

\begin{lemma}\label{lema:split-fine-enough}
	For box $\mathfrak{b}$ and its primary heap $H_\mathfrak{b}$, if the top element $\mathfrak{b}_{top}$ satisfies $len(\mathfrak{b}_{top})< \frac{2}{(2+\epsilon)d}len(\mathfrak{b})$, then  $rmax_\mathfrak{b}<\frac{2}{2+\epsilon} r_\mathfrak{b}$. 
\end{lemma} 

\begin{proof}
	It's sufficient to prove $\forall \mathfrak{b}'\in H_\mathfrak{b}, Est(\mathfrak{b}')<\frac{2}{2+\epsilon}Est(\mathfrak{b})$. 
	
	Based on Lemma \ref{lema:len-est}, we have $Est(\mathfrak{b}')\le d\cdot len(\mathfrak{b}')\le d\cdot len(\mathfrak{b}_{top})$ for any box $\mathfrak{b}'\in H_\mathfrak{b}$. Combining with the condition given in the lemma, we have $Est(\mathfrak{b}')< d\frac{2}{(2+\epsilon)d}len(\mathfrak{b})= \frac{2}{2+\epsilon} len(\mathfrak{b})$. And further combining with $len(\mathfrak{b})\le Est(\mathfrak{b})$, we finally have $Est(\mathfrak{b}')<\frac{2}{2+\epsilon}Est(\mathfrak{b})$, which proves the lemma.
	\qed
\end{proof}

The pseudo codes for building the box split tree are given in Algorithm \ref{alg:preprocess}. The algorithm also includes the invocation of Algorithm \ref{alg:nbr-maintain}, aimed to maintain the $Nbr$ sets, which will be introduced in the next section.

\begin{algorithm}[H]
	\caption{Preprocessing}\label{alg:preprocess}
	\KwIn{a point set $P$, and an approximation factor $\epsilon$}
	\KwOut{a box split tree $T$}
	
	\begin{spacing}{1.1}
		\tcp{Initialization}
	\end{spacing}
	Compute $\mathfrak{b}_0=MCB(P)$\;
	Compute the enclosing ball $B(c_{\mathfrak{b}_0},r_{\mathfrak{b}_0})$ of $\mathfrak{b}_0$\;
	Set $\mathfrak{b}_0$ to be the root of $T$\;
	Initialize the primary heap for $\mathfrak{b}_0$ with one key-value pair $(len(\mathfrak{b}_0),\mathfrak{b}_0)$\;
	Initialize the secondary heap $H_2$ with one key-value pair $(len(\mathfrak{b}_0),\mathfrak{b}_0)$\;
	
	\begin{spacing}{1.1}
		\tcp{Main loop}
	\end{spacing}	
	\While{$|\mathcal{B}|<n$}
	{
		Pop out the top element $\mathfrak{b}_{top}$ from $H_2$ and its corresponding primary heap $H_{\mathfrak{b}_\mathrm{s}}$\;
		Split $\mathfrak{b}_{top}$ by conducting one split step on $\mathfrak{b}$, generating $Succ(\mathfrak{b}_{top})$\;
		
		\ForEach{$\mathfrak{b}\in Succ(\mathfrak{b}_{top})$}{
			Add $\mathfrak{b}$ into $H_{\mathfrak{b}_\mathrm{s}}$, and maintain the heap\;
		}
		Let the current top element of $H_{\mathfrak{b}_\mathrm{s}}$ to be $\mathfrak{b}_t$\;
		Let $Flag=false$\;
		\uIf(\tcp*[h]{Applying Lemma \ref{lema:split-fine-enough} }){$len(\mathfrak{b}_t)<\frac{2}{(2+\epsilon)d}len(\mathfrak{b}_\mathrm{s})$}
		{
			Let $Flag=true$\;
			\ForEach{$\mathfrak{b}\in H_{\mathfrak{b}_\mathrm{s}}$}{
				Create a node and hang it under the node of $\mathfrak{b}_\mathrm{s}$\;
				Initialize the primary heap for $\mathfrak{b}$ with one key-value pair $(len(\mathfrak{b}),\mathfrak{b})$\;	
			}
		}
		\Else{
			Add $\mathfrak{b}_t$ into $H_2$.
		}
		Invoke Algorithm \ref{alg:nbr-maintain}, taking $\mathfrak{b}, Succ(\mathfrak{b})$, $rmax_{\mathcal{B}_{\mathfrak{b}_\mathrm{s}}}$, and the boolean value $Flag$ as the input of this invocation\;
	}
\end{algorithm}

\subsubsection{$Nbr$ sets maintaining}
Algorithm \ref{alg:nbr-maintain} for maintaining $Nbr(\mathfrak{b})$ is given below. It is invoked each time the main loop of Algorithm \ref{alg:preprocess} is executed, as shown above.

\begin{algorithm}[H]
	\caption{Maintaining $Nbr(b)$}\label{alg:nbr-maintain}
	\KwIn{box $\mathfrak{b}$, $Succ(\mathfrak{b})$, the neighbor range parameter $rmax_{\mathcal{B}_{\mathfrak{b}_\mathrm{s}}}$, and a boolean value $Flag$. }
	\ForEach{$\mathfrak{b}'\in Succ(\mathfrak{b})$}{
		$Nbr(\mathfrak{b}')\gets Nbr(\mathfrak{b})\cup Succ(\mathfrak{b})-\{\mathfrak{b}\}$\;
		Set $Est(\mathfrak{b}')$ according to Equation \ref{eqtn:est(b)}\;
		Update $rmax_{\mathfrak{b}_s}$\;\label{line:nbr-rmax_bs}
	}
	\ForEach{$\mathfrak{b}'\in Nbr(b)$}{
		\uIf{$Flag=ture$ and $\mathfrak{b}'$ is in a higher level that $\mathfrak{b}$}{
			$Nbr(\mathfrak{b}')\gets  Nbr(\mathfrak{b'})\cup Succ(\mathfrak{b})$
		}\Else{
			$Nbr(\mathfrak{b}')\gets  Nbr(\mathfrak{b'})\cup Succ(\mathfrak{b})-\{\mathfrak{b}\}$
		}
	}
	\ForEach{$\mathfrak{b}'\in Succ(\mathfrak{b})$}{
		\ForEach{$\mathfrak{b}''\in Nbr(b')$}{
			\If{$D(c_{\mathfrak{b}''},c_{\mathfrak{b}'})>(3+4/\epsilon)rmax_{\mathcal{B}_{\mathfrak{b}_\mathrm{s}}}$}{
				Delete $\mathfrak{b}''$ from $Nbr(\mathfrak{b}')$\;
				Delete $\mathfrak{b}'$ from $Nbr(\mathfrak{b}'')$\;
			}
		}
	}
\end{algorithm}

There are two parts of Algorithm \ref{alg:nbr-maintain} that need to be explained in detail.

The first is Line \ref{line:nbr-rmax_bs}. From Definition \ref{def:nbr} for $Nbr(\mathfrak{b})$, we can see that the maintaining of $Nbr(\mathfrak{b})$ is based on the value $rmax_{\mathcal{B}_{\mathfrak{b}_\mathrm{s}}}$ passed down by its super-box $\mathfrak{b}_\mathrm{s}$. In Algorithm \ref{alg:nbr-maintain}, Line \ref{line:nbr-rmax_bs} is aimed for updating $rmax_{\mathfrak{b}_\mathrm{s}}$ when the set of sub-boxes of $\mathfrak{b}_\mathrm{s}$ is changed. If $Nbr(\mathfrak{b})$ is implemented as a heap, then whenever any sub-box of $\mathfrak{b}_\mathrm{s}$ needs $rmax_{\mathcal{B}_{\mathfrak{b}_\mathrm{s}}}$, this value can be retrieved from the heap in constant time.

The other part is the second $foreach$ loop in Algorithm \ref{alg:nbr-maintain}. The functionality of the loop is explained in the following Lemma \ref{lema:nbr-propety}.

\begin{lemma}\label{lema:nbr-propety}
	The second $foreach$ loop ensures that for all box $\mathfrak{b}$ in the box split tree $T$, each $\mathfrak{b}'\in Nbr(\mathfrak{b})$ is either in the same level with $\mathfrak{b}$, or a degenerated box containing only one point.
\end{lemma}

\begin{proof}
	From Algorithm \ref{alg:preprocess}, we know that the boolean value $Flag$ indicates whether splitting $\mathfrak{b}$ causes the box split tree $T$ to grow. If $Flag=true$, $\mathfrak{b}$ becomes a inner node.
	In that case, if there is a box $\mathfrak{b}'\in Nbr(\mathfrak{b})$ where $\mathfrak{b}'$ is in a higher level than $\mathfrak{b}$, $\mathfrak{b}$ will remain in $Nbr(\mathfrak{b}')$ according to the second $foreach$ loop in Algorithm \ref{alg:nbr-maintain}. 
	First we claim that $\mathfrak{b}'$ is not an inner node. 
	If so, $\mathfrak{b}'$ must have been split before, and Algorithm \ref{alg:nbr-maintain} was invoked at the time $\mathfrak{b}'$ was split. 
	In this invocation, the $else$-branch of the second $foreach$ loop was executed, and $\mathfrak{b}'$ was already deleted from $Nbr(\mathfrak{b})$. 
	This conflicts with $\mathfrak{b}'\in Nbr(\mathfrak{b})$. 
	Thus, we get the conclusion that $\mathfrak{b}'$ is not an inner node and never has been split. 
	
	We count on the next several invocations of Algorithm \ref{alg:nbr-maintain} to delete $\mathfrak{b}'$ from $Nbr(\mathfrak{b})$.
	After $\mathfrak{b}$ is split, $\mathfrak{b}'$ may be split but the boxes in $Succ(\mathfrak{b}')$ may be in the same level with $\mathfrak{b}'$, which only introduces more higher-level boxes into $Nbr(\mathfrak{b})$. 
	The critical time is when $Succ(\mathfrak{b}')$ is in the next level of $\mathfrak{b}'$. 
	In that case, While Algorithm \ref{alg:nbr-maintain} is invoked by splitting $\mathfrak{b}'$, the $else$-branch of the second $foreach$ loop will delete $\mathfrak{b}'$ from $Nbr(\mathfrak{b})$. 
	Of course, if $\mathfrak{b}'$ contains only one point and can not be split, it will remain in $Nbr(\mathfrak{b})$ until Algorithm \ref{alg:preprocess} terminates. 

	So far we have eliminated a box in $Nbr(\mathfrak{b})$ containing more than one point and in a higher level than $\mathfrak{b}$. Repeatedly applying the same proof, we will eliminate all such box in $\mathfrak{b}$. Then the claim is proved.
\qed
\end{proof}

\subsection{Query}
The query algorithm goes down the tree $T$ returned by Algorithm \ref{alg:preprocess} level by level. At each level of $T$, the algorithm $\mathcal{A}$ for $(c,r)$-NN will be invoked, and the input parameters of $\mathcal{A}$ are set according to Lemma \ref{lema:crnn-param}. The pseudo codes are given in Algorithm \ref{alg:query}.

\begin{algorithm}[h]
	\caption{Query}\label{alg:query}
	\SetKw{Break}{break}
	\SetKw{Continue}{continue}
	\KwIn{query point $q$, data set $P$, box split tree $T$, and algorithm $\mathcal{A}$ for $(c,r)$-NN}
	\KwOut{$\epsilon$-NN of $q$ in $P$}
	
	set $\mathfrak{b}_c=root(T)$\;
	\If{$D(q,c_{\mathfrak{b}_c}\ge T_2(\mathfrak{b}))$}{
		pick any point $p'\in \mathfrak{b}_c\cap P$\;
		\Return $p'$\; \label{line:query:return1}	
	}
	\While{$|\mathfrak{b}_c|>1$}{
		$\mathcal{B}_c\gets Nbr(\mathfrak{b}_c)$\;
		$P_c\gets\bigcup\limits_{\mathfrak{b}\in \mathcal{B}_c}{\mathfrak{b}\cap P}$\;\label{line:query:assign-Pc}
		
		invoke $\mathcal{A}$, where the input of $\mathcal{A}$ is set according to Lemma \ref{lema:crnn-param}\;
		\uIf{the query returns NO}{
			pick any point $p'\in P_c$\;
			\Return $p'$\; \label{line:query:return}
		}
		\Else(\tcp*[h]{the query returns the center $c_{\mathfrak{b}'}$ of box $\mathfrak{b}'$}){
				set $\mathfrak{b}_c=\mathfrak{b}'$\;\label{line:query:reassign-bc}
				\Continue \;	
		}
	}
	$P_c\gets Nbr(\mathfrak{b}_c)\cap P$\; \label{line:query:final-search}
	Conduct brute-force search in $P_c$ to find the exact NN\; \label{line:query:termination}	
\end{algorithm}

We should spend some efforts to explain the termination condition in Algorithm \ref{alg:query}. First we introduce a lemma about $Nbr(\mathfrak{b})$ when $|\mathfrak{b}|=1$.

\begin{lemma}\label{lema:termination}
	For a box $\mathfrak{b}$ satisfying $|\mathfrak{b}|=1$, all the boxes $\mathfrak{b}'\in Nbr(\mathfrak{b})$ contain only one point, i.e., $|\mathfrak{b}'|=1$.
\end{lemma}

\begin{proof}
	According to Algorithm \ref{alg:nbr-maintain}, if a box $\mathfrak{b}$ satisfies $|\mathfrak{b}|=1$, the algorithm will  keep updating $Nbr(\mathfrak{b})$ until Algorithm \ref{alg:nbr-maintain} is not invoked any more. And that is when Algorithm \ref{alg:preprocess} terminates and when all the boxes degenerate and contain only one point. It implies that any box $\mathfrak{b}'\in Nbr(\mathfrak{b})$ satisfies $|\mathfrak{b}'|=1$.
	\qed
\end{proof}

According to the above lemma, when the WHILE loop breaks, all boxes in $\mathfrak{b}_c\cup Nbr(\mathfrak{b}_c)$ contains only one point. The brute-force search takes $O(|Nbr(\mathfrak{b}_c)|)$ time. We will bound this complexity in the next section.

\section{Analysis}\label{sec:analysis}

\subsection{Correctness}
First we prove the correctness of our query algorithm by introduce the following lemma \ref{lema:nn-in-Pc}.

\begin{lemma}\label{lema:nn-in-Pc}
	In every execution of the loop body, Algorithm \ref{alg:query} ensures that the exact nearest neighbor $p^*\in P_c$ after the assignment of $P_c$(Line \ref{line:query:assign-Pc}).
\end{lemma}

\begin{proof}
	The proof proceeds by induction. At the beginning of Algorithm \ref{alg:query}, apparently we have $P_c=P$, and $p^*\in P_c$ holds trivially. As inductive hypothesis, we assume $p^*\in P_c$ after Line \ref{line:query:assign-Pc} in one execution of the loop body. Consider the rest of the loop body. If Line \ref{line:query:return} is executed, the algorithm will return, and the induction finishes. 
	If Line \ref{line:query:reassign-bc} is executed, Lemma \ref{lema:nn-in-nbr} ensures that $p^*\in Nbr(\mathfrak{b}')=Nbr(\mathfrak{b}_c)$. 
	Thus, in the next execution of the loop body, the reassignment of $P_c$ at Line \ref{line:query:assign-Pc} makes $p^*\in P_c$ to hold again. Then by mathematical induction, the lemma is proved.
	\qed
\end{proof}

\begin{theorem}[Correctness]\label{thrm:correctness}
	The point $p'$ returned by Algorithm \ref{alg:query} is an $\epsilon$-NN to $q$ in $P$, i.e., if $p^*$ is the exact NN to $q$ in $P$, then $D(q,p')\le (1+\epsilon) D(q,p^*)$.
\end{theorem}

\begin{proof}
	Considering Algorithm \ref{alg:query}, if it returns at \ref{line:query:return1}, Lemma \ref{lema:t2} ensures that the picked point $p'$ is an $\epsilon$-NN to q; if it returns at Line \ref{line:query:return}, Lemma \ref{lema:crnn-param} ensures that the point $p'$ returned here is an $\epsilon$-NN to $q$; and if the algorithm finally goes out of the WHILE loop and executes brute force search in the final $P_c$ assigned at Line \ref{line:query:final-search}, Lemma \ref{lema:nn-in-Pc} ensures that the exact nearest neighbor lies in $P_c$, and thus the brute-force search returns the exact nearest neighbor to $q$ for sure.
	\qed
\end{proof}

\subsection{Complexities}
Before we bound the complexity of our algorithm, we should first bound the size of $Nbr(\mathfrak{b})$ for any box $\mathfrak{b}$ by introducing a lemma from \cite{Vaidya1986}.

\begin{lemma}[\cite{Vaidya1986}]\label{lema:packing}
	Let $r$ be a positive number. During the execution of the split method described in Section \ref{subsec:preprocessing}, at each time before splitting a box, let $\mathcal{B}$ be the current box collection, and let $\mathfrak{b}_L$ be the box with the largest volume in $\mathcal{B}$. For any box $\mathfrak{b}\in \mathcal{B}$, the size of the set 
	$\{\mathfrak{b}'\in\mathcal{B} \mid {D_{min}(\mathfrak{b},\mathfrak{b}')\le r\cdot Est(\mathfrak{b}_L) }\}$ is at most $2^d (2d\lceil r\rceil+3)^d$.
\end{lemma}

Based on the lemma above, we can bound the size of $Nbr(\mathfrak{b})$ for any box $\mathfrak{b}$ in the box split tree $T$ constructed in Algorithm \ref{alg:preprocess}.
\begin{lemma}\label{lema:nbr-size}
	The size of $Nbr(\mathfrak{b})$ defined in Definition \ref{def:nbr} and constructed in Algorithm \ref{alg:nbr-maintain} is $O((\frac{d}{\epsilon})^d)$.
\end{lemma}

\begin{proof}
	We prove this by using Lemma \ref{lema:packing}.
	
	By Definition \ref{def:nbr}, $Nbr(\mathfrak{b})= \left\{ \mathfrak{b}'\mid {D(c_{\mathfrak{b}'},c_\mathfrak{b})\le(3+4\epsilon)rmax_{\mathcal{B}_{\mathfrak{b}_\mathrm{s}}}} \right\}$, where $\mathcal{B}_{\mathfrak{b}_{\mathrm{s}}}= Nbr(\mathfrak{b}_\mathrm{s})$ and $\mathfrak{b}_{\mathrm{s}}$ is the super box of $\mathfrak{b}$. On the other hand, Lemma \ref{lema:packing} concerns the set $\left\{\mathfrak{b}'\in \mathcal{B}\mid{D_{min}(\mathfrak{b}', \mathfrak{b})\le r\cdot Est(\mathfrak{b}_L)} \right\}$, where $\mathfrak{b}_L$ is the box with the largest volume in the box collection $\mathcal{B}$. We should fill the gap between 
	$Nbr(\mathfrak{b})$ and 
	$\left\{\mathfrak{b}\in \mathcal{B}'\mid{D_{min}(\mathfrak{b}', \mathfrak{b})\le r\cdot Est(\mathfrak{b}_L)} \right\}$ by considering two relationships: 1) $D(c_{\mathfrak{b}'}, c_\mathfrak{b})$ and $D_{min}(\mathfrak{b}',\mathfrak{b})$, and 2) $rmax_{\mathcal{B}_{\mathfrak{b}_\mathrm{s}}}$ and $Est(\mathfrak{b}_L)$.
	
	$1)$ Since the center of the enclosing ball of box $\mathfrak{b}$ is one arbitrary point inside $\mathfrak{b}$, we have $D_{min}(\mathfrak{b},\mathfrak{b}')\le D(c_\mathfrak{b},c_{\mathfrak{b}'})$. 
	Thus, it can be easily verified that:
	$$\forall K>0, \left\{\mathfrak{b}'\mid{D(c_\mathfrak{b},c_{\mathfrak{b}'})\le K}\right\} \subseteq \left\{\mathfrak{b}'\mid {D_{min(\mathfrak{b},\mathfrak{b}')} \le K} \right\}$$ .
	
	$2)$ Recall Definition \ref{def:rmax-collection}, $rmax_{\mathcal{B}_{\mathfrak{b}_\mathrm{s}}}=\max\limits_{\mathfrak{b}'\in \mathcal{B}_{\mathfrak{b}_\mathrm{s}}}{rmax_{\mathfrak{b}'}}=\max\limits_{\mathfrak{b}'\in \mathcal{B}_{\mathfrak{b}_\mathrm{s}}}{ \max\limits_{\mathfrak{b}''\in Chd(\mathfrak{b}')}{Est(\mathfrak{b}'')}}$ where $Chd(\mathfrak{b}')$ is the set of sub-boxes of $\mathfrak{b}'$. Let $\mathcal{B}'=\bigcup_{\mathfrak{b}'\in \mathcal{B}_{\mathfrak{b}_\mathrm{s}}}{Chd(\mathfrak{b}')}$, and $\mathcal{B}'$ is clearly a subset of the whole box collection $\mathcal{B}$.
	On the other hand, $\mathfrak{b}_L$ is the box with the largest volume in $\mathcal{B}$.
	Based on Fact \ref{lema:len-est}, we have $rmax_{\mathcal{B}_{\mathfrak{b}_\mathrm{s}}}=\max\limits_{b'\in \mathcal{B}'}\{Est(\mathfrak{b}')\}\le \max\limits_{b'\in \mathcal{B}'}\{d\cdot len(\mathfrak{b}')\}\le d\cdot len(\mathfrak{b}_L)\le d\cdot Est(\mathfrak{b}_L)$. 
	Thus we have 
	$$\forall \alpha>0, \left\{ \mathfrak{b}'\mid {D_{min}(\mathfrak{b},\mathfrak{b}')} \le \alpha\cdot rmax_{\mathcal{B}_{\mathfrak{b}_\mathrm{s}}} \right\} \subseteq \left\{ \mathfrak{b}'\mid {D_{min}(\mathfrak{b},\mathfrak{b}')} \le \alpha\cdot d\cdot Est(\mathfrak{b}_L)\right\}$$ 
	
	Combining $1)$ and $2)$, we have:
	$$\left\{\mathfrak{b}'\mid{D(c_\mathfrak{b},c_{\mathfrak{b}'}) \le (3+4/\epsilon)rmax_{\mathcal{B}_{\mathfrak{b}_\mathrm{s}}}}\right\} \subseteq \left\{ \mathfrak{b}'\mid {D_{min}(\mathfrak{b},\mathfrak{b}')} \le (3+4/\epsilon)d\cdot Est(\mathfrak{b}_L)\right\}$$
	
	Then according to Lemma \ref{lema:packing}, we have $$\left|\left\{ b'\mid {D_{min}(\mathfrak{b},\mathfrak{b}')} \le (3+4/\epsilon)d\cdot Est(\mathfrak{b}_L)\right\}\right|\le 2^d\left(2d\lceil (3+4/\epsilon)d \rceil+3\right)^d = O((\frac{d}{\epsilon})^d)$$

	Note that $Nbr(\mathfrak{b}_\mathrm{s})$ may be updated and the value of $rmax_{\mathcal{B}_{\mathfrak{b}_\mathrm{s}}}$ may be changed while $Nbr(\mathfrak{b})$ is being maintained based on an older value of $rmax_{\mathcal{B}_{\mathfrak{b}_\mathrm{s}}}$.
	But it does not influent the result in this lemma, because Lemma \ref{lema:packing} ensures that the size of the set considered in the lemma is bounded \emph{every time} before a box is split. Thus, even though $Nbr(\mathfrak{b})$ may be maintained based on an older value of $rmax_{\mathcal{B}_{\mathfrak{b}_\mathrm{s}}}$, $|Nbr(\mathfrak{b})|=O((\frac{d}{\epsilon})^d)$ still holds.
	\qed
	
\end{proof}

We introduce and prove another lemma which is about the property of the box split tree $T$ constructed in preprocessing phase.

\begin{lemma}\label{lema:split-tree}
	For a point set $P$ where $|P|=n$, the fully built split tree $T$ constructed based on $P$ has the following properties:
	\begin{enumerate}
		\item\label{lema:split-tree:number}
		There are at most $2n$ nodes in $T$.
		\item\label{lema:split-tree:build-time}
		The total time to build $T$ is $O(dn\log{n})$.
	\end{enumerate}
\end{lemma}

\begin{proof}
	For the first statement, the proof starts from the following two observations.
	
	1) There are exactly $n$ leaf nodes in $T$. Because $T$ is fully built, each box at leaf node contains only one point. 
	
	2) Each node has at least $2$ child nodes and at most $|P|=n$ child nodes. This comes from the definition of box split tree.

	Combining the two observations, if $T$ is a full binary tree, then $T$ has at most $2n$ nodes, which can be easily verified. As long as one node has three child nodes or more, the total number of nodes would be less than $2n$. The extreme situation is that the root has $n$ child nodes and there are totally $n+1$ nodes in $T$. So we can conclude that there are at most $2n$ nodes in $T$.
	
	For the second statement, we can divide the time to build $T$ into two parts. One is the total time to conduct all the split steps. The other is the total time to manipulate the primary and secondary heaps. We analyze the time complexity of the two parts as follows.
	
	$1)$ The total time to conduct the split steps is $O(dn\log{n})$. This is already proved in \cite{Vaidya1986}. We omit the proof and refer the readers to \cite{Vaidya1986} for the details.
	
	$2)$ We have proved that there are at most $2n$ boxes in the fully built split tree $T$. Considering the manipulation of the primary and secondary heaps, it is easily verified that each box $\mathfrak{b}$ may exist in at most two heaps, i.e., one primary and one secondary. For each heap, $\mathfrak{b}$ can only be pushed into it only once, and be popped out of it only once. The number of the boxes in the heap is at most $2n$, so for each box $\mathfrak{b}$, the heap manipulation time  incurred by $\mathfrak{b}$ is $O(\log{n})$. Thus, the total time to manipulate the primary and secondary heaps is $O(n\log{n})$.
	
	Adding the two parts of complexity, we conclude that the total time to build $T$ is $O(dn\log{n})$.
	\qed
\end{proof}

Now we start to prove the complexities of our algorithm, including preprocessing time, space and query time complexities.

\begin{theorem}[Preprocessing Time Complexity]\label{thrm:preprocessing-complexity}
	The complexity of Algorithm \ref{alg:preprocess} for preprocessing is $O(O((\frac{d}{\epsilon})^d\cdot n\log{n}))$.
\end{theorem}

\begin{proof}
	The complexity of Algorithm \ref{alg:preprocess} can be divided into two parts, namely, (1) the total time to build the box split tree $T$, and (2) the total time to maintain $Nbr$ data structures and $Est$ value for all boxes. The first part of complexity is already proved in Lemma \ref{lema:split-tree}. Here we prove the second part.
	
	Our Algorithm \ref{alg:nbr-maintain} is very similar to an algorithm for maintaining the $Nbr$ set in \cite{Vaidya1986}.  We prove the complexity of Algorithm \ref{alg:preprocess} by similar techniques in \cite{Vaidya1986}.
	If the $Nbr(\mathfrak{b})$ sets is implemented by a heap, which allows insertion and deletion in $O(\log{n})$ time, and allow access to largest value of $D_{min}(\mathfrak{b}, \mathfrak{b}'), \mathfrak{b'}\in Nbr(\mathfrak{b})$ in $O(\log{n})$ time, then we can use the similar analysis in \cite{Vaidya1986} and bound the time to maintain $Nbr$ sets and $Est$ values by $O((\frac{d}{\epsilon})^d\cdot n\log{n})$. 
	The details are omitted.
	
	In summary, we have proved the desired preprocessing complexity.
	\qed
\end{proof}

\begin{theorem}[Space Complexity]\label{thrm:space-complexity}
	The space complexity of Algorithm \ref{alg:preprocess} is $O((\frac{d}{\epsilon})^d\cdot n)$.
\end{theorem}

\begin{proof}
	The space complexity of the algorithm is bounded by the number of boxes in the tree $T$ multiplying the size of $Nbr(\mathfrak{b})$ sets maintained for each box. According to Lemma \ref{lema:split-tree}, there are at most $2n$ boxes in $T$. And according to Lemma \ref{lema:nbr-size}, $|Nbr(\mathfrak{b})|= O((\frac{d}{\epsilon})^d$. Multiplying the two factors, we get the desired result.
	\qed
\end{proof}

\begin{theorem}[Query Time Complexity]\label{thrm:query-complexity}
	Algorithm \ref{alg:query} invokes $O(\log{n})$ times of the algorithm $\mathcal{A}$ for $(c,r)$-NN problem.
\end{theorem}

\begin{proof}
	Considering the box split tree returned by Algorithm \ref{alg:preprocess}, its fan-out, i.e., the number of child nodes of a node, is at least 2. This comes from the definition of box split tree. And the number of leaf nodes is $n$, so that the height of the tree is $O(\log{n})$. Further, by Lemma \ref{lema:nbr-propety} we know that all boxes in any box $\mathfrak{b}$ in $T$ are in the same level with $\mathfrak{b}$. Thus, Algorithm \ref{alg:query} invokes at most one $(c,r)$-NN query at each level of the tree. Hence, the number of invoked $(c,r)$-NN queries is $O(\log{n})$.
	\qed
\end{proof}

\section{Conclusion}
In this paper we proposed a new algorithm for reducing $\epsilon$-NN problem to $(c,r)$-NN problem. Compared to the former works for the same reduction problem, our algorithm achieves the lowest query time complexity, which is $O(\log{n})$ times of invocations of the algorithm for $(c,r)$-NN problem. We elaborately designed the input parameters of each of the invocation, and built a dedicated data structure in preprocessing phase to support the query procedure. A box split method proposed in \cite{Vaidya1986} is used as a building block for the algorithm of preprocessing phase. Our paper also raises a problem which is to reduce the exponential complexity on $d$ introduced by the box split method. This is left as our future work.

%
%
%
%

\bibliographystyle{splncs04}
\bibliography{library}
\end{document}